\documentclass{article}

\usepackage{fullpage}
\usepackage{amsmath}
\usepackage{graphicx}%
\usepackage{amsfonts}%
\usepackage{amssymb}
\usepackage{amsthm}
\usepackage{cite}

\newtheorem{theorem}{Theorem}

\newtheorem{corollary}[theorem]{Corollary}

\newtheorem{definition}[theorem]{Definition}
\newtheorem{example}[theorem]{Example}

\newtheorem{lemma}[theorem]{Lemma}

\newtheorem{proposition}[theorem]{Proposition}

\begin{document}

\title{On the Quasi-Balanceable Class of Linear Quantum Stochastic Systems 
\footnote{This research was supported by the Australian Research Council}}

\author{Onvaree~Techakesari and Hendra~I.~Nurdin
\thanks{O. Techakesari and H. I. Nurdin are with the School of Electrical Engineering and Telecommunications, UNSW Australia, Sydney NSW 2052, Australia. Email: o.techakesari@unsw.edu.au \& h.nurdin@unsw.edu.au}}

\maketitle

\begin{abstract}
This paper concerns the recently proposed quasi-balanced truncation model reduction method for linear quantum stochastic systems.
It has previously been shown that the quasi-balanceable class of systems (i.e. systems that can be truncated via the quasi-balanced method) includes the class of completely passive systems. In this work, we refine the previously established characterization of quasi-balanceable systems  and show that the class of quasi-balanceable systems is strictly larger than the class of completely passive systems. In particular, we derive a novel characterization of completely passive linear quantum stochastic systems solely in terms of the controllability Gramian of such systems.
Exploiting this result, we prove that all linear quantum stochastic systems with a pure Gaussian steady-state (active systems included)  are all quasi-balanceable, and establish a new complete parameterization for this important class of systems. Examples are provided to illustrate our results.
\end{abstract}

\noindent{\bf Keywords:} Linear quantum stochastic systems, model reduction, symplectic transformations, quantum optical systems, open Markov quantum systems


\section{Introduction} \label{sec:intro}

Over the past two decades, quantum feedback control has garnered a great deal of interest among physicists and control theorists. This active research area is motivated by the promising benefits of the quantum mechanical features that are present in systems such as those in quantum optics and optomechanical systems; for example, see \cite{WM10,BE08,HM12,CTSAM13}. 
 A key aspect in quantum feedback control theory is the modeling of the quantum systems under consideration. In particular, linear quantum stochastic differential equations (QSDE) have been consistently used to describe open quantum harmonic oscillators coupled with external quantum fields, giving rise to the class of linear quantum stochastic systems \cite{BE08,JNP06,GJN10,NJD08,Nurd10a,Nurd10b,GZ13}. This class can be viewed as a quantum analogue of classical (non-quantum) linear stochastic systems and are prominent in  fields such as quantum optics, optomechanics, and superconducting circuits. They model devices ranging from optical cavities, mesoscopic mechanical resonators, optical and superconducting parametric amplifiers, to linear quantum memories such as  gradient echo quantum memories (GEM); see, e.g.,  \cite{GZ04,HM12,KAKKCSL13,HCHJ13}. 
Several applications of these systems have been reported in the literature. To mention a few, they can serve as coherent feedback controllers \cite{JNP06,NJP07b} to cool optomechanical systems \cite{HM12}, modify the characteristics of squeezed light produced by an OPO \cite{CTSAM13}, and reshape the dynamics of an electromechanical circuit \cite{KAKKCSL13}. As optical filters, they can, for instance, modify the wavepacket shape of single \cite{ZJ13} and multi-photon \cite{Zhang14} sources. Also,  Gaussian cluster states that are of interest for continuous-variable one way quantum computers \cite{MLGWRN06} can be generated dissipatively using linear quantum stochastic systems \cite{KY12}.  Besides for quantum information technologies, linear quantum stochastic systems are also of interest for classical signal processing on quantum devices, for instance as light processors when photons, rather than electrons, transport information between cores on a chip and between chips \cite{BKSWW07}. In combination with proposed nonlinear ultra-low power optical logic gates, such as \cite{Mab11b}, they form building blocks for ultra-efficient all-optical classical information processing circuits.

Linear quantum stochastic systems can be completely described (in the Heisenberg picture) by a quartet of matrices $(A,B,C,D)$ analogous to linear systems in classical control theory. However, unlike classical (non-quantum) systems, the matrices $(A,B,C,D)$ of a linear quantum stochastic system cannot be arbitrary and must satisfy a set of conditions for the system to be physically realizable \cite{JNP06}.
In \cite{NJD08,Nurd10a,Nurd10b}, a network synthesis theory for physically realizable linear quantum stochastic systems has been presented as a quantum analogue of network synthesis theory for linear electrical networks. 
Unfortunately, a linear quantum stochastic system may have many degrees of freedom 
(i.e., the system dimension is large) making it difficult to synthesize.
In this situation, a physically realizable approximating (reduced order) model of the system, that leads to a small modeling error, is of practical interest. Moreover, acquisition of an appropriately reduced dynamical system model allows us to examine (in a computationally feasible manner) the physical phenomena of the system and, hence, provides a means to design an appropriate controller. 
The contributions of this paper are related to this model reduction problem.

In \cite{BvHS07,GNW10,P13}, singular perturbation approximation methods have been studied for model reduction of quantum systems, whilst an eigenvalue truncation method for linear quantum stochastic systems has been presented in \cite{Pet12}. 
In a recent work \cite{Nurd13b}, the quasi-balanced truncation method has been proposed for linear quantum stochastic systems as an adaptation of the well-known balanced truncation method for classical linear systems. As with its classical counterpart, the quasi-balanced method truncates modes corresponding to subsystems with the smallest Hankel singular values, while preserving the physical realizability property in the reduced order approximation. A bound on the model approximation error, in terms of a bound on the $H^{\infty}$ norm of the error transfer function, is also provided. However, in contrast with the classical case, it is not possible to apply  quasi-balanced truncation to generic linear quantum stochastic systems, and importantly, a complete physical description of all linear quantum stochastic systems that belong to this class has not been established. The reason for the former is that the allowable similarity transformations for quasi-balancing are symplectic transformations that preserve canonical commutation relations of internal degrees of freedom of the system and  guarantee that the reduced system is again physically realizable. A consequence of the restriction to a symplectic transformation is that it is not always possible to transform an arbitrary linear quantum stochastic system to one in which the controllability and observability Gramians are simultaneously diagonal.

In this paper, we present a new, more explicit, characterization of quasi-balanceable linear quantum stochastic systems (i.e. systems that can be truncated using the quasi-balanced method), refining a necessary and sufficient commutator condition obtained earlier in \cite{Nurd13b}. We show that all asymptotically stable linear quantum stochastic systems with  a pure Gaussian steady-state, as parameterized in \cite{KY12,Yama12}, are quasi-balanceable. Additionally,  we obtain a new complete parameterization for this class of system. This illustrates that the class of quasi-balanceable systems is strictly larger than the class of completely passive linear quantum stochastic systems as established in \cite{Nurd13b}. Additionally, we also construct a theoretical example of a quasi-balanceable system that does not have a pure Gaussian steady-state. 

The remainder of this paper is structured as follows. In Section \ref{sec:prelim}, we define the class of linear quantum stochastic systems under consideration. The new characterization of quasi-balanceable systems is provided in Section \ref{sec:characterization}.
Section \ref{sec:Gaussian} presents an application of the new characterization of quasi-balanceable systems to linear quantum stochastic systems that have a pure Gaussian steady-state. Some concluding remarks are then provided in Section \ref{sec:conclusion}.

\section{Preliminaries} \label{sec:prelim}

\subsection{Notation}
We will use the following notation: $\imath=\sqrt{-1}$, $^*$ denotes the adjoint of a linear operator
as well as the conjugate of a complex number. If $A=[a_{jk}]$ then $A^{\#}=[a_{jk}^*]$, and $A^{\dag}=(A^{\#})^{\top}$, where $(\cdot)^{\top}$ denotes matrix transposition.  $\Re\{A\}=(A+A^{\#})/2$ and $\Im\{A\}=\frac{1}{2\imath}(A-A^{\#})$.
We denote the identity matrix by $I$ whenever its size can be
inferred from context and use $I_{n}$ to denote an $n \times n$
identity matrix. Similarly, $0_{m \times n}$ denotes  a $m \times n$ matrix with zero
entries but drop the subscript when its dimension  can be determined from context. We use
${\rm diag}(M_1,M_2,\ldots,M_n)$  to denote a block diagonal matrix with
square matrices $M_1,M_2,\ldots,M_n$ on its diagonal, and ${\rm diag}_{n}(M)$
denotes a block diagonal matrix with the
square matrix $M$ appearing on its diagonal blocks $n$ times. Also, we
will let $\mathbb{J}=\left[\begin{array}{rr}0 & 1\\-1&0\end{array}\right]$ and $\mathbb{J}_n=I_{n} \otimes \mathbb{J}={\rm diag}_n(\mathbb{J})$.

\subsection{The class of linear quantum stochastic systems}
Let $x=(q_1,p_1,q_2,p_2,\ldots,q_n,p_n)^\top$ denote a vector of  the canonical position and
momentum operators of a {\em many degrees of freedom quantum
harmonic oscillator} satisfying the canonical commutation
relations (CCR)  $xx^\top -(xx^\top )^\top =2\imath \mathbb{J}_n$.
A {\em linear quantum stochastic
system} \cite{BE08,JNP06,NJD08} $G$ is a quantum system defined by three {\em parameters}:
(i) A quadratic Hamiltonian $H=\frac{1}{2}  x^\top  R x$ with $R=R^\top  \in
\mathbb{R}^{2n \times 2n}$, (ii) a coupling operator $L=Kx$, where $K$ is
an $m \times 2n$ complex matrix, and (iii) a unitary $m \times m$
scattering matrix $S$. For shorthand, we write $G=(S,L,H)$ or $G=(S,Kx,\frac{1}{2}  x^\top Rx)$. The
time evolution $x(t)$ of $x$ in the Heisenberg picture ($ t
\geq 0$) is given by the quantum stochastic differential equation
(QSDE) (see \cite{BE08,JNP06,NJD08}):
\begin{align}
dx(t) &=  A_0x(t)dt+ B_0\left[\small \begin{array}{c} d\mathcal{A}(t)
\\ d\mathcal{A}(t)^{\#} \end{array} \normalsize \right];  
x(0)=x, \notag\\
dY(t) &=  C_0 x(t)dt+  D_0 d\mathcal{A}(t), \label{eq:qsde-out}
\end{align}
with $A_0=2\mathbb{J}_n(R+\Im\{K^{\dag}K\})$, $B_0=2\imath \mathbb{J}_n [\begin{array}{cc}
-K^{\dag}S & K^\top S^{\#}\end{array}]$,
$C_0=K$, and $D_0=S$. Here $Y(t)=(Y_1(t),\ldots,Y_m(t))^{\top}$ is a vector of
continuous-mode bosonic {\em output fields} that results from the interaction of the quantum
harmonic oscillators and the incoming continuous-mode bosonic quantum fields in the $m$-dimensional vector
$\mathcal{A}(t)$. Note that the dynamics of $x(t)$ is linear, and
$Y(t)$ depends linearly on $x(s)$, $0 \leq s \leq t$. We refer to $n$ as the {\em
  degrees of freedom} of the system or, more simply, the {\em degree} of the system.
    
Following \cite{JNP06}, it will be  convenient to write the dynamics in quadrature form as
\begin{align}
dx(t)&=Ax(t)dt+Bdw(t);\, x(0)=x. \nonumber\\
dy(t)&= C x(t)dt+ D dw(t), \label{eq:qsde-out-quad}
\end{align}
with
\begin{eqnarray*}
w(t)&=&2(\Re\{\mathcal{A}_1(t)\},\Im\{\mathcal{A} _1(t)\},\Re\{\mathcal{A} _2(t)\},\Im\{\mathcal{A} _2(t)\},\ldots,\Re\{\mathcal{A} _m(t)\},\Im\{\mathcal{A} _m(t)\})^{\top}; \\
y(t)&=&2(\Re\{Y_1(t)\},\Im\{Y_1(t)\},\Re\{Y_2(t)\},\Im\{Y_2(t)\}, \ldots,\Re\{Y_m(t)\},\Im\{Y_m(t)\})^{\top}.
\end{eqnarray*}
Here, the real matrices $A,B,C,D$ are in a one-to-one correspondence
with 
$A_0,B_0,C_0,D_0$. Also, $w(t)$ is taken to be in a vacuum state where it
satisfies the It\^{o} relationship $dw(t)dw(t)^{\top} = (I+\imath \mathbb{J}_m)dt$; see \cite{JNP06}. Note that in this form it follows that $D$ is a real unitary symplectic matrix. That is, it is both unitary (i.e., $DD^{\top}=D^{\top} D=I$) and symplectic (a real $m \times m$ matrix is symplectic if $D \mathbb{J}_m D^{\top} =\mathbb{J}_m$). However, in the most general case, $D$ can be generalized to a symplectic matrix that represents a quantum network that includes ideal squeezing devices acting on the incoming field $w(t)$ before interacting with the system \cite{GJN10,NJD08}.  
In general one may not be interested in all outputs of the system but only in a subset of them, see, e.g., \cite{JNP06}. That is, one is often only interested in certain pairs of the output field quadratures in $y(t)$. Thus, in the most general scenario, $y(t)$ can have an even dimension $n_y <2m$.
The matrices $A$, $B$, $C$, $D$ of a linear quantum stochastic system cannot be arbitrary and are not independent of one another. In fact, for the system to be physically realizable \cite{JNP06,NJD08}, meaning it represents a meaningful physical system, they must satisfy the constraints (see \cite{JNP06,NJD08,Nurd13b}):
\begin{eqnarray}
&&A\mathbb{J}_n + \mathbb{J}_n A^{\top} + B\mathbb{J}_mB^{\top}=0, \label{eq:pr-1}\\
&& \mathbb{J}_n  C^{\top} +B\mathbb{J}_mD^{\top}=0, \label{eq:pr-2}\\
&&D \mathbb{J}_m D^{\top} =  \mathbb{J}_{n_y/2}. \label{eq:pr-3}
\end{eqnarray}

\section{Model reduction by quasi-balanced truncation} \label{sec:characterization}

In the work \cite{Nurd13b}, the quasi-balanced truncation model reduction method was proposed to provide an avenue to approximate a physically realizable linear quantum stochastic system with a smaller (i.e. reduced ordered) physically realizable system. It is essentially an adaptation to linear quantum stochastic systems of the classical balanced truncation method.
This method involves the controllability and observability Gramians of the system.
For an asymptotically stable system ($A$ is Hurwitz), the controllability and observability Gramians, $P=P^{\top}\geq 0$ and $Q=Q^{\top}\geq 0$ ($P>0$ if system is controllable and $Q>0$ if it is observable), respectively, are the unique solutions to the Lyapunov equations:
\begin{align}
 AP + PA^\top + BB^\top &= 0 \label{eq:lyap-P} \\
 QA + A^\top Q + C^\top C &= 0. \label{eq:lyap-Q}
\end{align}
In this quasi-balancing method, the system is symplectically transformed to a similar system with diagonal controllability and observability Gramians.
Pairs of quadratures corresponding to the smallest entries of the geometric mean of the product of the diagonal Gramians (which are also pairs corresponding to the smallest Hankel singular values of the system) are then truncated. In \cite{Nurd13b}, it is also shown that $H^{\infty}$ truncation error bounds can be established for quasi-balanced truncation, analogous to the classical case. Let us now re-state the important result on the existence of a symplectic similarity transformation $T$ required in applying the quasi-balanced truncation method.
\begin{theorem} \label{thm:quasi-balanced} 
\cite[Theorem 8]{Nurd13b}
There exists a symplectic matrix $T$ such that $TPT^\top = \Sigma_P$, $T^{-\top}QT^{-1} = \Sigma_Q$, 
with $\Sigma_X$ $(X \in \{P,Q\})$ of the form $\Sigma_X = {\rm diag}(\sigma_{X,1}I_2,\sigma_{X,2}I_2,\ldots,\sigma_{X,n}I_2)$
with $\sigma_{X,1},\sigma_{X,2},\ldots,\sigma_{X,n} \geq 0$
the symplectic eigenvalues of $X$ (symplectic eigenvalues of $P$ need not be the same as those of Q), if and only if
$[\mathbb{J}_n P, Q\mathbb{J}_n] = 0$.
\end{theorem}

In the remainder of the paper, we will say that a linear quantum stochastic system has a quasi-balanced realization or the system is quasi-balanceable if there exists a real symplectic transformation matrix $T$ such that the transformed system $\tilde{G}=(TAT^{-1},TB,CT^{-1},D)$ has diagonal controllability and observability Gramians of the special form described in Theorem \ref{thm:quasi-balanced}.
In the sequel, we will present some new results which provide further characterization of quasi-balanceable systems.
However, before establishing the new results, let us present a useful proposition.
First, in this paper we say that a matrix $A \in \mathbb{R}^{2n \times 2n}$ is skew-Hamiltonian if $A^\top \mathbb{J}_n = \mathbb{J}_n A$. 

\begin{proposition} \label{prop:skew-Hamiltonian}
A real symmetric matrix $A = A^\top  \in \mathbb{R}^{2n \times 2n}$ is skew-Hamiltonian if and only if $A$ has the block form
$A = [A_{jk}]_{j,k = 1,2,\ldots,n}$ with
\begin{equation}
 A_{jk}
 = \left\{\begin{array}{cl}
    \hat{a}_{jk} I_{2}; & \mbox{if } j = k \\
    \left[\begin{array}{rr}
     a_{jk,1} & a_{jk,2} \\ -a_{jk,2} & a_{jk,1}
    \end{array}\right] & \mbox{otherwise}
   \end{array}\right.
 \label{eq:skew-Hamiltonian}
\end{equation}
where $\hat{a}_{jk}, a_{jk,1}, a_{jk,2} \in \mathbb{R}$ for all $j,k = 1,2,\ldots,n$
with $a_{jk,1} = a_{kj,1}$ and $a_{jk,2} = -a_{kj,2}$.
\end{proposition}
\begin{proof}
The proposition follows from direct inspection of $A\mathbb{J}_n = \mathbb{J}_n A^\top$ 
(i.e. the definition of skew-Hamiltonian matrix).
\end{proof}

\subsection{Completely passive systems}

In \cite{Nurd13b}, it has been shown that the class of completely passive linear quantum stochastic systems are quasi-balanceable.
Completely passive systems are those which can be synthesized using only passive components; for example in quantum optics these components include optical cavities and beam splitters. 
Some characterization of completely passive systems have been presented in \cite{Nurd10b}.
In particular, it has been shown that, for a completely passive system, 
the Hamiltonian matrix $R$ is of the form \eqref{eq:skew-Hamiltonian} 
(that is, by Proposition \ref{prop:skew-Hamiltonian}, the matrix is skew-Hamiltonian),
the coupling matrix $K$ has the special form
\begin{equation}
 K = \left[\begin{array}{ccccccc} M_1 & \imath M_1 & M_2 & \imath M_2 & \cdots & M_n & \imath M_n \end{array} \right]
 \label{eq:K-passive}
\end{equation}
where $M_j \in \mathbb{C}^m$ is a column vector for all $j = 1,2,\ldots,n$,
and the matrix $D$ is unitary symplectic or there exists a real matrix $E \in \mathbb{R}^{(2m-n_y)\times 2m}$ such that 
$[D^\top \ E^\top]^\top$ is unitary symplectic (the latter is when only a subset of quadrature pairs of outputs is considered \cite[Section 2.3]{Nurd13b}). 
We now present results characterizing completely passive systems in terms of the controllability Gramian $P$.
The following partial characterization of completely passive systems in terms of $P$ was established in \cite{Nurd13b}.

\begin{lemma} \cite[Theorem 13]{Nurd13b} \label{lemma:cp-1}
For any completely passive linear quantum stochastic system that is asymptotically stable (i.e. the matrix $A$ is Hurwitz), $P = I$.
\end{lemma}

Now, we show that the converse is also true.

\begin{lemma} \label{lemma:cp-2}
For any asymptotically stable linear quantum stochastic system $G=(A,B,C,D)$,
with a unitary symplectic $D$ or there exists a real matrix $E$ such that $[D^\top \ E^\top]^\top$ is unitary symplectic,
the controllability Gramian $P = I$ only if the system is completely passive.
\end{lemma}
\begin{proof}
Note that the required property on $D$  ensures that there is a  unitary matrix $S \in \mathbb{C}^{m \times m}$ such that the identities  $A=2\mathbb{J}_n(R+\Im\{K^{\dag}K\})$ and $B=2\imath \mathbb{J}_n [\begin{array}{cc}
-K^{\dag}S & K^\top S^{\#}\end{array}]\Gamma$, where $\Gamma$ is a complex matrix satisfying $\Gamma \Gamma^{\dag}=\frac{1}{2}I$ \cite{JNP06}, hold for some $R=R^{\top} \in \mathbb{R}^{2n \times 2n}$ and $K \in \mathbb{C}^{m \times 2n}$, when the system is physically realizable. 

Since $P = I$, the Lyapunov equation for the controllability Gramian becomes $A + A^\top + BB^\top = 0$
which gives that
\begin{equation}
 2(\mathbb{J}_n R - R\mathbb{J}_n) + 2(\mathbb{J}_n\Im\{K^\dagger K\} + \Im\{K^\dagger K\}\mathbb{J}_n)
 + BB^\top = 0. \label{eq:cp-2-lyap2}
\end{equation}
Let $r_{jk}$, $u_{jk}$, $v_{jk}$, and $b_{jk}$ denote the $jk$-th element of the matrix $R$, $K$, $\Im\{K^\dagger K\}$, and $BB^\top$, respectively.
By examining each $2 \times 2$ block of the above equation, we have that
\begin{align}
 4\left(r_{(2j-1)(2j)}-v_{(2j-1)(2j)}\right) + b_{(2j-1)(2j-1)} = 0 \label{eq:cp-2-diag-1} \\
 -4\left(r_{(2j-1)(2j)}+v_{(2j-1)(2j)}\right) + b_{(2j)(2j)} = 0 \label{eq:cp-2-diag-2}
\end{align}
for all $j = 1,2,\ldots,n$, where 
$v_{jk}$ and $b_{jk}$ are given by
\begin{align*}
 v_{(2j-1)(2j)} &= \sum_{\ell=1}^{m} \bigg( \Re\{u_{\ell (2j-1)}\}\Im\{u_{\ell (2j)}\} 
                   - \Re\{u_{\ell (2j)}\}\Im\{u_{\ell (2j-1)}\}\bigg) \\
 b_{(2j-1)(2j-1)} &= 4\sum_{\ell=1}^{m} \bigg( \Re\{u_{\ell (2j)}\}^2 + \Im\{u_{\ell (2j)}\}^2 \bigg) \\
 b_{(2j)(2j)} &= 4\sum_{\ell=1}^{m} \bigg( \Re\{u_{\ell (2j-1)}\}^2 + \Im\{u_{\ell (2j-1)}\}^2 \bigg).
\end{align*}
The second and third identities above follow from the identity 
$BB^{\top}=-2\mathbb{J}_n (K^{\dag}K+K^{\top}K^{\#})\mathbb{J}_n$ while the latter follows from the identity
$B=2\imath \mathbb{J}_n[\begin{array}{cc}-K^{\dag} S & K^{\top}S^{\#}\end{array}]\Gamma$.

By adding \eqref{eq:cp-2-diag-1} and \eqref{eq:cp-2-diag-2}, we have that
\begin{align}
 0 = \sum_{\ell=1}^{m} \left[ \bigg( \Re\{u_{\ell (2j-1)}\}  - \Im\{u_{\ell (2j)}\}  \bigg)^2 
     + \bigg( \Re\{u_{\ell (2j)}\} - \Im\{u_{\ell (2j-1)}\}\bigg)^2\right].
\end{align}
This implies that $u_{\ell (2j)} = iu_{\ell (2j-1)}$ for all $j \in [1,n]$ and all $\ell \in [1,m]$.
Hence, $K$ is of the form \eqref{eq:K-passive} i.e. $K$ is that of a completely passive system.

Since $K$ is of the form \eqref{eq:K-passive}, it follows that $\Re\{K\} = \Im\{K\}\mathbb{J}_n^\top$.
Now note that $\Re\{K^\dagger K\} = \Re\{K\}^\top \Re\{K\} + \Im\{K\}^\top \Im\{K\}$
and 
$\Im\{K^\dagger K\} = \Re\{K\}^\top \Im\{K\} - \Im\{K\}^\top \Re\{K\}$.
Thus, we have that 
$\mathbb{J}_n \Im\{K^\dagger K\} = \Im\{K^\dagger K\} \mathbb{J}_n$ \\ $= - \mathbb{J}_n \Re\{K^\dagger K\} \mathbb{J}_n^\top$.
Also, we highlight that $BB^\top = 4\mathbb{J}_n \Re\{K^\dagger K\} \mathbb{J}_n^\top$.
By substituting these into \eqref{eq:cp-2-lyap2}, we have that $\mathbb{J}_n R = R \mathbb{J}_n$.
Therefore, $R$ is skew-Hamiltonian i.e. $R$ is that of a completely passive system. Since $D$ is of the form required for a completely passive system, this establishes the lemma statement.
\end{proof}

It is noted that Lemma \ref{lemma:cp-2} completes the parameterization of completely passive linear quantum stochastic systems in terms of the controllability Gramian $P$.
Apart from providing further characterization for completely passive systems, 
the above results will be used in the sequel to provide further insight into quasi-balanceable systems. 

\subsection{Characterization of quasi-balanceable linear quantum \\ stochastic systems}

In this section, we present novel results that provide a more refined characterization of linear quantum stochastic systems possessing a quasi-balanced realization.
Firstly, from Lemmas \ref{lemma:cp-1} and \ref{lemma:cp-2}, we will show that the class of linear systems symplectically similar to the class of completely passive systems are quasi-balanceable.

\begin{theorem} \label{thm:quasi-transformed-passive}
Let $G = (A,B,C,D)$ be an asymptotically stable and completely passive linear quantum stochastic system\footnote{Note that for a completely passive system, asymptotic stability is equivalent to minimality \cite{GZ13}.}.
Then a linear quantum stochastic system of the form \\
 $\tilde{G} = (T_0AT_0^{-1},T_0B,CT_0^{-1},D)$, with $T_0$ a real symplectic matrix, is quasi-balanceable.
In this case, the quasi-balancing transformation  for the system $\tilde{G}$ is  $T = \tilde{T}T_0^{-1}$,
where $\tilde{T}$ is a unitary symplectic matrix (which can be obtained by applying Theorem 7 of \cite{Nurd13b}) such that 
$$
\tilde T^{-\top}Q \tilde T^{-1} = {\rm diag}(\sigma_{Q,1}I_2,\sigma_{Q,2}I_2,\ldots,\sigma_{Q,n}I_2),
$$
and $Q$ is the observability Gramian of the completely passive system $G$. 
\end{theorem}
\begin{proof}
First, we note that the system $\tilde{G}$ can be symplectically transformed into its symplectically similar completely passive system $G$ via the transformation matrix $T_0^{-1}$.
The result of this theorem then follows from \cite[Theorem 3]{Nurd13b}.
\end{proof}

The above theorem shows that any system which is symplectically similar to a completely passive linear quantum stochastic system  
is quasi-balanceable. 
Therefore, this result implies that the class of quasi-balanceable systems is larger than the class of completely passive systems as established in \cite{Nurd13b}. 


Let us now present the useful definition of a $\Sigma$-unitary matrix (as a generalization of the notion of a unitary matrix)
before presenting further characterization of quasi-balanceable systems.
 
\begin{definition}
For any diagonal matrix $\Sigma > 0$ of the form 
$\Sigma = {\rm diag}(\sigma_{1}I_2,\sigma_{2}I_2,\ldots,\sigma_{n}I_2)$,
a matrix $T$ is said to be $\Sigma$-unitary if
\begin{equation}
 T\Sigma T^\top = \Sigma \quad \mbox{or, equivalently,} \quad T^\top \Sigma T = \Sigma.
\end{equation}
\end{definition}

\begin{theorem} \label{thm:quasi-balanced-characterization1}
For any asymptotically stable linear quantum stochastic system $G = (A,B,C,D)$ 
with a diagonal controllability Gramian 
$\Sigma_P = {\rm diag}(\sigma_{P,1}I_2,\sigma_{P,2}I_2,\ldots,\sigma_{P,n}I_2) > 0$,
the system has a quasi-balanced realization if and only if $Q$ has a block form $Q = [Q_{jk}]_{j,k = 1,2,\ldots,n}$ with
\begin{equation}
 Q_{jk}
 = \left\{\begin{array}{cl}
    \hat{q}_{jk} I_{2}; & \mbox{if } j = k \\
    \delta_{\sigma_{P,j} \sigma_{P,k}}
    \left[\begin{array}{rr}
     q_{jk,1} & q_{jk,2} \\ -q_{jk,2} & q_{jk,1}
    \end{array}\right] & \mbox{otherwise}
   \end{array}\right.
 \label{eq:Q-form}
\end{equation}
where $\hat{q}_{jk}, q_{jk,1}, q_{jk,2} \in \mathbb{R}$ for all $j,k = 1,2,\ldots,n$
with $q_{jk,1} = q_{kj,1}$ and $q_{jk,2} = -q_{kj,2}$.
Here, $\delta_{ab}$ is a generalized Kronecker delta defined as
\begin{equation*}
 \delta_{ab} 
 \triangleq \left\{\begin{array}{cl}
    1; & \mbox{if } a = b \\
    0; & \mbox{otherwise}
   \end{array}\right..
\end{equation*} 
\end{theorem}
\begin{proof}
For the ``if'' part, suppose that $Q$ has the form \eqref{eq:Q-form}.
Using Proposition \ref{prop:skew-Hamiltonian}, we have that $\Sigma_P$ and $Q$ are skew-Hamiltonian.
Also, from direct inspection, $\Sigma_P Q = Q \Sigma_P$ (i.e. the two matrices commute with each other).
Then we have that
$[\mathbb{J}_n \Sigma_P, Q\mathbb{J}_n] = 0$.
The ``if'' part of the result then follows from Theorem \ref{thm:quasi-balanced}.

For the ``only if'' part, suppose that the system has a quasi-balanced realization.
From Theorem \ref{thm:quasi-balanced}, this implies that there exists a symplectic $\Sigma_P$-unitary matrix $T$ such that
$T^{-\top} Q T^{-1} = \Sigma_Q$
where $\Sigma_Q = {\rm diag}(\sigma_{Q,1}I_2,\sigma_{Q,2}I_2,\ldots,\sigma_{Q,n}I_2) \geq 0$.
Moreover, from the definition of a $\Sigma_P$-unitary matrix, we have that $T^{-1}\Sigma_P T^{-\top} = \Sigma_P$,
$T^{-1}\Sigma_P^{-1} T^{-\top} = \Sigma_P^{-1}$.
Thus we can write
\begin{align}
 T^{-\top} Q T^{-1} &= \Sigma_Q \nonumber \\
 (T^\top \mathbb{J}_n \Sigma_P^{-1}) T^{-\top} Q T^{-1} (\Sigma_P \mathbb{J}_n T)
   &= (T^\top \mathbb{J}_n \Sigma_P^{-1}) \Sigma_Q (\Sigma_P \mathbb{J}_n T) \nonumber \\
 T^\top \mathbb{J}_n T T^{-1} \Sigma_P^{-1} T^{-\top} Q T^{-1} \Sigma_P T^{-\top} T^\top \mathbb{J}_n T 
    &= -T^\top \Sigma_Q T \nonumber \\
 \mathbb{J}_n \Sigma_P^{-1} Q \Sigma_P \mathbb{J}_n &= -Q.
\end{align}

We now examine each $2 \times 2$ block of $\mathbb{J}_n \Sigma_P^{-1} Q \Sigma_P \mathbb{J}_n = -Q$.
Let $q_{jk}$ denote the $jk$-th element of the matrix $Q$.
Examining the diagonal blocks, we have for all $j = 1,2,\ldots,n$ that
\begin{equation}
 \left[\begin{array}{cc}
  q_{(2j-1)(2j-1)} & q_{(2j-1)(2j)} \\
  q_{(2j-1)(2j)} & q_{(2j)(2j)}
 \end{array}\right]
 =
 \left[\begin{array}{cc}
  q_{(2j)(2j)} & -q_{(2j-1)(2j)} \\
  -q_{(2j-1)(2j)} & q_{(2j-1)(2j-1)}
 \end{array}\right].
\end{equation}
This implies that $q_{(2j-1)(2j-1)} = q_{(2j)(2j)}$ and $q_{(2j-1)(2j)} = 0$ for all $j = 1,2,\ldots,n$.
This establishes \eqref{eq:Q-form} when $j = k$.

Now, examining the off-diagonal blocks of $\mathbb{J}_n \Sigma_P^{-1} Q \Sigma_P \mathbb{J}_n = -Q$, 
we have for all $j,k = 1,2,\ldots,n$ that
\begin{equation}
 \left[\hspace{-0.2em}\begin{array}{cc}
  q_{(2j-1)(2k-1)} & q_{(2j-1)(2k)} \\
  q_{(2j)(2k-1)} & q_{(2j)(2k)}
 \end{array}\hspace{-0.2em}\right]
 =
 \frac{\sigma_{P,k}}{\sigma_{P,j}}
 \left[\hspace{-0.2em}\begin{array}{cc}
  q_{(2j)(2k)} & -q_{(2j)(2k-1)} \\
  -q_{(2j-1)(2k)} & q_{(2j-1)(2k-1)}
 \end{array}\hspace{-0.2em}\right].
\end{equation}
From the above equation, it follows that all elements of the matrix of the left hand side is $0$ if $\sigma_{P,j} \neq \sigma_{P,k}$.
When $\sigma_{P,j} = \sigma_{P,k}$, it follows from the above equation that
$q_{(2j)(2k)} = q_{(2j-1)(2k-1)}$ and $q_{(2j)(2k-1)} = -q_{(2j-1)(2k)}$.
This establishes \eqref{eq:Q-form} when $j \neq k$.
\end{proof}

We note that a similar result to Theorem \ref{thm:quasi-balanced-characterization1} can be obtained for an asymptotically stable system with diagonal observability Gramian $\Sigma_Q = {\rm diag}(\sigma_{Q,1}I_2,\sigma_{Q,2}I_2,\ldots,\sigma_{Q,n}I_2)>0$.
That is, such system has a quasi-balanced realization if and only if the controllability Gramian $P$ is of the form \eqref{eq:Q-form} (in this case $P$ can be allowed to be positive semidefinite). 

\begin{corollary} \label{cor:quasi-balanced-characterization2}
An asymptotically stable system $G = (A,B,C,D)$, with controllability Gramian $P = P^\top  > 0$ 
and observability Gramian $Q = Q^\top  \geq 0$,
has a quasi-balanced realization if and only if 
$\tilde{Q} = T_0^{-\top} Q T_0^{-1}$ has the block form \eqref{eq:Q-form},
where $T_0$ is a symplectic matrix such that $T_0 P T_0^\top = \Sigma_P$
and \\ $\Sigma_P = {\rm diag}(\sigma_{P,1}I_2,\sigma_{P,2}I_2,\ldots,\sigma_{P,n}I_2)$.
\end{corollary}
\begin{proof}
The ``if'' part follows from Theorem \ref{thm:quasi-balanced-characterization1}.
For the ``only if'' part, suppose that the system has a quasi-balanced realization. 
Then there exists a symplectic matrix $T$ such that
$TPT^\top = \Sigma_P$ and $T^{-\top}QT^{-1} = \Sigma_Q = {\rm diag}(\sigma_{Q,1}I_2,\sigma_{Q,2}I_2,\ldots,\sigma_{Q,n}I_2)$.
Let $T_0$ be another symplectic matrix such that $T_0PT_0^\top = \Sigma_P$
and let $T = \tilde{T}T_0$, where $\tilde{T}$ is a symplectic $\Sigma_P$-unitary matrix. By Theorem \ref{thm:quasi-balanced-characterization1}, a symplectic $\Sigma_P$-unitary matrix $\tilde{T}$ exists only if $T_0^{-\top} Q T_0^{-1}$ has the form \eqref{eq:Q-form}. This establishes the ``only if'' part of the theorem statement.
\end{proof}

The above corollary provides further insight into the characteristics of quasi-balanceable systems. It includes as  special cases the completely passive linear quantum stochastic systems and systems of the form described in Theorem \ref{thm:quasi-transformed-passive}. However, the corollary goes beyond these two classes and we now present a theoretical example of a quasi-balanceable system which is not completely passive and also outside of the systems covered by Theorem \ref{thm:quasi-transformed-passive}. A treatment of a distinguished class of systems that meet the requirement of Theorem  \ref{thm:quasi-transformed-passive} will be given in the next section.

\begin{example}
Consider a $n$ degree of freedom, $2m$ input and $2p$ output linear quantum stochastic system with a Hamiltonian matrix $R$ of the passive form \eqref{eq:skew-Hamiltonian} and the coupling matrix $K = [K_1^\top \ K_2^\top]^\top$ where $K_1 \in \mathbb{C}^{p \times 2n}$ is of the passive form \eqref{eq:K-passive} and $K_2 \in \mathbb{C}^{(m-p) \times 2n}$ is a dispersive coupling of the form:
\begin{align*}
 K_2 = \left[\begin{array}{ccccccc}
        u_{11} & 0 & u_{12} & 0 & \cdots & u_{1n} & 0 \\
        0 & u_{11} & 0 & u_{12} & \cdots & 0 & u_{1n} \\
        \vdots & \vdots & \vdots & \vdots & \vdots & \vdots & \vdots \\
        u_{k1} & 0 & u_{k2} & 0 & \cdots & u_{kn} & 0 \\
        0 & u_{k1} & 0 & u_{k2} & \cdots & 0 & u_{kn}
       \end{array}\right]
\end{align*}
where $k = m-p$ and $u_{ij} \in \mathbb{C}$ for all $i = 1,2,\ldots,k$ and $j = 1,2,\ldots,n$.
In other words, this illustrative system is described by a completely passive Hamiltonian matrix $R$ driven by $m$ vacuum fields.
The $n$ oscillator modes of the system are {\it passively} coupled to the first $p$ fields and are dispersively coupled to the remaining $m-p$ other fields (i.e., each oscillator mode is coupled to only one quadrature of each of the last $m-p$ fields).
The outputs of the system are then taken from those first $p$ {\it passively} coupled fields.

For this system, it can be shown that the controllability Gramian $P$ is skew-Hamiltonian and the observability Gramian is $Q = I$.
By applying Theorem \ref{thm:quasi-balanced-characterization1}, the system has a quasi-balanced realization.
This example theoretically illustrates that the class of quasi-balanceable systems is strictly larger than the class of completely passive systems and its symplectic transformations.
\end{example}

\section{Application: A new complete parameterization of all linear quantum stochastic systems generating an arbitrary pure Gaussian steady state} \label{sec:Gaussian}

An important class of linear quantum stochastic systems, especially in quantum optics and quantum information,
is the class of Gaussian systems 
e.g. linear quantum stochastic systems initialized in a Gaussian state,
or asymptotically stable linear quantum stochastic systems that generate a Gaussian state of the internal oscillators in steady-state (the limit that $t \rightarrow \infty$), irrespective of whether or not the initial joint state of the oscillators is Gaussian. 
A Gaussian state can be completely characterized by its first and second order statistics.
Particularly, the second order statistics of a Gaussian state is described by a covariance matrix $P$ which contains important information such as purity of the state and its symplectic spectrum (that can be used to determine entanglement in bipartite systems) \cite{L05}.
For an asymptotically stable system (having a unique steady state) driven by bosonic fields in the vacuum state, 
the steady-state covariance matrix $P$ is the unique solution to the Lyapunov equation \eqref{eq:lyap-P}.

We now show that quasi-balanceable systems are closely connected to systems with a pure Gaussian steady state which play an important role in continuous-variable quantum information technologies; see \cite{KY12,Yama12} for example.
It has been shown that for a pure Gaussian state,
the covariance matrix (or the controllability matrix) can be written as (see, e.g., \cite{L05}) 
\begin{equation}
 P = TT^\top \label{eq:P-pure-state}
\end{equation} 
where $T$ is a symplectic matrix. 
A characterization of linear Gaussian systems with pure steady states have been presented in \cite{KY12,Yama12}.
Following the equation above, we now establish a key connection between this class of Gaussian systems with pure steady-state and the class of completely passive systems.
We will use $\nu_{P,1},\nu_{P,2},\ldots,\nu_{P,n}$ to denote the $n$ largest positive eigenvalues of $i\mathbb{J}_n P$.
These positive eigenvalues are referred to as {\it symplectic eigenvalues} of $P$ (where $P>0$).

\begin{lemma} \label{lemma:cp2pure}
Any asymptotically stable linear quantum stochastic system with a pure steady state has controllability 
Gramian $P$ with symplectic eigenvalues $\nu_{P,j} = 1$ for all $j = 1,2,\ldots,n$.
Moreover, if $G=(A,B,C,D)$ is an asymptotically stable and completely passive Gaussian system, 
then the transformed system $\tilde{G} = (TAT^{-1},TB,CT^{-1},D)$, with $T$ symplectic, has a pure Gaussian steady state.
Conversely, if $G=(A,B,C,D)$ is an asymptotically stable Gaussian system with a pure steady state, and $D$ is unitary symplectic or there exists a real matrix $E$ such that $[D^\top \ E^\top]^\top$ is unitary symplectic,
then there exists a symplectic matrix $T$ such that $\tilde{G} = (TAT^{-1},TB,CT^{-1},D)$ is a completely passive system.
\end{lemma}
\begin{proof}
From \eqref{eq:P-pure-state} and Williamson's Theorem \cite{Will36,PSL09}, it follows that the covariance matrix $P$, of any asymptotically stable Gaussian system with pure steady state, has symplectic eigenvalues $\nu_{P,j} = 1$ for all $j = 1,2,\ldots,n$.
Now, for the transformed system $\tilde{G}$, the controllability Gramian $\tilde{P} = TPT^\top$.
The lemma statement then follows from \eqref{eq:P-pure-state}, Lemma \ref{lemma:cp-1}, and Lemma \ref{lemma:cp-2}.
%
\end{proof}

We emphasize that the above lemma provides a new complete parameterization of linear quantum stochastic Gaussian systems with pure steady states that differs from those presented in \cite{KY12,Yama12}, when $D$ is of the stated form (i.e., the form of $D$ for a completely passive system).  The new parameterization is given in terms of a sympletic matrix $T$ and the parameters of a completely passive linear quantum stochastic system. Note that, by using different means, \cite[Section 3]{Yama12} already showed that a system with a pure Gaussian steady state can be symplectically transformed to a completely passive system and that its transfer function is all- pass. Here, we present a short and straightforward proof of this by exploiting a property of the controllability Gramian of completely passive systems given in Lemma \ref{lemma:cp-2}. The converse, that every completely passive system gives rise to a system with a pure Gaussian steady state via symplectic transformations, has not been established previously and it is shown here via Lemma \ref{lemma:cp2pure}. Thus, it was not recognized in \cite{Yama12} that the relationship between completely passive systems and systems with a pure Gaussian steady state gives rise to a simpler alternative complete parameterization of the latter class of systems. This parameterization provides an insight into the preparation of pure Gaussian steady states which play an important role in quantum information and computation \cite{Yama12}.


\begin{corollary} \label{cor:quasi-balanced-pureGaussian}
Any asymptotically stable Gaussian system $G = (A,B,C,D)$ with a pure steady state has a quasi-balanced realization.
Let $P>0$ and $Q>0$ be the controllability and observability Gramians of the system $G$, respectively.
Then the quasi-balanced transformation $T = T_Q T_P$ is a symplectic matrix where
$T_P$ is a symplectic matrix such that $T_P P T_P^\top = I$ and 
$T_Q$ is a unitary symplectic matrix such that $T_Q^{-\top} T_P^{-\top} Q T_P^{-1} T_Q^{-1} = \Sigma_Q$.
Here, $\Sigma_Q = {\rm diag}(\sigma_{Q,1}I_2,\sigma_{Q,2}I_2,\ldots,\sigma_{Q,n}I_2)$.
\end{corollary}
\begin{proof}
The proof of this corollary follows directly from Lemma \ref{lemma:cp2pure} and Theorem \ref{thm:quasi-transformed-passive}.
\end{proof}

We now present an example illustrating the application of the above result.

\begin{example} \label{ex:Gaussian}
Consider a two-mode optical parametric oscillator (OPO), originally considered in \cite{KY12}.
which is driven by two vacuum fields.
This optical parametric oscillator is described by a 2 degree of freedom, 4 input, and 2 output linear quantum stochastic system 
with the following Hamiltonian matrix $R$ and coupling matrix $K$:
\begin{align*}
 R = \left[\begin{array}{cccc}
      0 & \epsilon_1 & 0 & -\gamma \\
      \epsilon_1 & 0 & \gamma & 0 \\
      0 & \gamma & 0 & \epsilon_2 \\
      -\gamma & 0 & \epsilon_2 & 0 
     \end{array}\right], 
 \quad
 K = \sqrt{\gamma} 
     \left[\begin{array}{cccc}
      1 & i & 1 & i \\ 1 & i & 1 & i
     \end{array}\right]
\end{align*}
where $\epsilon_1, \epsilon_2 \in \mathbb{R}$ are the squeezing parameters and $\gamma$ is the decay rate.
The system is asymptotically stable ($A$ is Hurwitz) when $|\epsilon_j| < \gamma$ for all $j = 1,2$.
Using Lemmas \ref{lemma:cp-1} and \ref{lemma:cp-2} in combination with \eqref{eq:P-pure-state}, 
it can be shown that the system has pure steady-state when $\epsilon_1 + \epsilon_2 = 0$
(which coincides with the result of \cite{KY12}).
From Lemma \ref{lemma:cp2pure}, this implies that the system can be symplectically transformed into a completely passive system.
Moreover, using the logarithmic negativity measure $E_N$ for bipartite Gaussian systems, see, e.g., \cite[Eq. (27)]{L05} for its definition, the above system produces entanglement with $E_N = 0.2925$ (although the completely passive system that is similar to this system does not generate entanglement).

Let $\epsilon_1 = 1\times 10^6$, $\epsilon_2 = -1\times 10^6$ and $\gamma = 5\times 10^6$ Hz.
This two-mode OPO can be described by a linear quantum stochastic system with the matrices:
\begin{align*}
 A &= 10^6\left[\begin{array}{cccc}
       -18 & 0 & -10 & 0 \\
       0 & -22 & 0 & -10 \\
       -30 & 0 & -22 & 0 \\
       0 & -30 & 0 & -18 
      \end{array}\right], \\
 B &= -4.47\times 10^3\left[\begin{array}{cc}
       I_2 & I_2 \\ I_2 & I_2
      \end{array}\right], \\ 
 C &= 4.47\times 10^3 \left[\begin{array}{cc}
        I_2 & I_2
       \end{array}\right], \\
 D &= \left[\begin{array}{cc}
        I_2 & 0_{2\times 2}
      \end{array}\right].
\end{align*}
The controllability and observability Gramians $P$ and $Q$ are as follows:
\begin{align*}
 P &= \left[\begin{array}{cccc}
       1.250 & 0 & -0.250 & 0 \\
       0 & 0.833 & 0 & 0.167 \\
       -0.250 & 0 & 1.250 & 0 \\
       0 & 0.167 & 0 & 0.833 
      \end{array}\right],  \\
 Q &= \left[\begin{array}{cccc}
       0.417 & 0 & 0.083 & 0 \\
       0 & 0.625 & 0 & -0.125 \\
       0.083 & 0 & 0.417 & 0 \\
       0 & -0.125 & 0 & 0.625 
      \end{array}\right].
\end{align*}
By applying Corollary \ref{cor:quasi-balanced-pureGaussian}, we have that the above system has a quasi-balanced realization
and the quasi-balancing transformation matrix $T = T_QT_P$ where $T_Q = I$ and
\begin{equation*}
 T_P = \left[\begin{array}{cccc}
       0 & -0.8660 & 0 & 0.8660 \\
       0.5773 & 0 & -0.5774 & 0 \\
       0 & 0.7071 & 0 & 0.7071 \\
       -0.7071 & 0 & -0.7071 & 0
       \end{array}\right].
\end{equation*}
Using the transformation matrix $T$, the two-mode OPO is symplectically transformed to a completely passive linear quantum stochastic system $\tilde{G} = (\tilde{A},\tilde{B},\tilde{C},D)$ where
\begin{align*}
 \tilde{A} &= TAT^{-1} 
   = 10^6\left[\begin{array}{rr}
       0_{2\times 2} & -9.80 I_2 \\
       9.80 I_2 & -40 I_2
      \end{array}\right], \\
 \tilde{B} &= TB
   = -6.32\times 10^3\left[\begin{array}{cc}
       0_{2\times 2} & 0_{2\times 2} \\ \mathbb{J} & \mathbb{J}
      \end{array}\right], \\ 
 \tilde{C} &= CT^{-1}
   = -6.32\times 10^3 \left[\begin{array}{cc}
        0_{2\times 2} & \mathbb{J}
       \end{array}\right].
\end{align*}
The completely passive system can also be described by the following Hamiltonian matrix $\tilde{R} = T^{-\top}RT^{-1}$ and coupling matrix $\tilde{K} = KT^{-1}$ i.e., 
\begin{align*}
 \tilde{R} = 10^6\left[\begin{array}{cccc}
       0 & 0 & 0 & 4.90 \\
       0 & 0 & -4.90 & 0 \\
       0 & -4.90 & 0 & 0 \\
       4.90 & 0 & 0 & 0 \\
      \end{array}\right], \quad
\tilde{K} = 10^3 \left[\begin{array}{cccc}
        0 & 0 & 3.16\imath & -3.16 \\
        0 & 0 & 3.16\imath & -3.16
       \end{array}\right].
\end{align*}
\end{example}

\section{Conclusion} \label{sec:conclusion}

In this work, we 
have presented a more explicit characterization of quasi-balanceable linear quantum stochastic systems than the one previously established in \cite{Nurd13b}.
Previously, quasi-balanceable systems has been shown to include the class of completely passive systems. Here we have shown that quasi-balanceable systems  in fact form a class of systems that is strictly larger than the class of completely passive systems. 
We also established a new complete parameterization of completely passive linear quantum stochastic systems solely in terms of the controllability Gramian. Furthermore, as a by-product of our results, a connection between completely passive systems and linear quantum stochastic systems with a pure Gaussian steady-state is presented. This connection has allowed us to provide a novel characterization of linear quantum stochastic systems dissipatively generating a pure Gaussian steady state, different from the one previously obtained in \cite{KY12}. In fact, we showed that all linear quantum stochastic systems producing a pure Gaussian steady state are symplectic transformations of completely passive linear quantum stochastic systems. Subsequently, all such systems are quasi-balanceable.

\bibliographystyle{IEEEtran} 
\bibliography{rip,otr}

\end{document}